\documentclass[letterpaper, 10 pt, conference]{ieeeconf}  %

\IEEEoverridecommandlockouts                              %
\usepackage{tikz}
\usepackage{subcaption}
\usepackage{booktabs}
\usepackage{url}
\usepackage[hidelinks]{hyperref}
\usepackage{cite}
\usepackage{amsmath,amssymb,amsfonts}
\usepackage{algorithmic}
\usepackage{graphicx}
\usepackage{textcomp}
\usepackage{xcolor}
\def\BibTeX{{\rm B\kern-.05em{\sc i\kern-.025em b}\kern-.08em
    T\kern-.1667em\lower.7ex\hbox{E}\kern-.125emX}}
\newtheorem{thm}{Theorem }%
\newtheorem{proposition}{Proposition}%
\newtheorem{cor}{Corollary}%
\newtheorem{lemma}{Lemma}%
\newtheorem{defn}{Definition}%
\newtheorem{rem}{Remark }%
\newtheorem{example}{Example}[section]

\usepackage[ruled,vlined,linesnumbered]{algorithm2e}
\usepackage{makecell}

\newcommand{\set}[1]{\left\{#1\right\}}

\newcommand{\ra}{\rightarrow}
\newcommand{\Real}{\mathbb{R}}
\newcommand{\eps}{\varepsilon}

\renewcommand{\subset}{\subseteq}

\newcommand{\C}{\mathcal{C}}

\title{\LARGE \bf
Computing Control Lyapunov-Barrier Functions:\\ Softmax Relaxation and Smooth Patching with Formal Guarantees
}

\author{Jun Liu and Maxwell Fitzsimmons  %
\thanks{This research was supported in part by the Natural Sciences and Engineering Research Council of Canada and the Canada Research Chairs program.}%
\thanks{Jun Liu and Maxwell Fitzsimmons are with the Department of Applied Mathematics, University of Waterloo, Waterloo, Ontario N2L 3G1, Canada.  Email: \texttt{j.liu@uwaterloo.ca (Jun Liu)}
        }%
}

\begin{document}

\maketitle
\thispagestyle{empty}
\pagestyle{empty}

\begin{abstract}
We present a computational framework for synthesizing a single smooth Lyapunov function that certifies both asymptotic stability and safety. We show that the existence of a strictly compatible pair of control barrier and control Lyapunov functions (CBF-CLF) guarantees the existence of such a function on the exact safe set certified by the barrier. To maximize the certifiable safe domain while retaining differentiability, we employ a log-sum-exp (softmax) relaxation of the nonsmooth maximum barrier, together with a counterexample-guided refinement that inserts half-space cuts until a strict barrier condition is verifiable. We then patch the softmax barrier with a CLF via an explicit smooth bump construction, which is always feasible under the strict compatibility condition. All conditions are formally verified using a satisfiability modulo theories (SMT) solver, enabled by a reformulation of Farkas’ lemma for encoding strict compatibility. On benchmark systems, including a power converter, we show that the certified safe stabilization regions obtained with the proposed approach are often less conservative than those achieved by state-of-the-art sum-of-squares (SOS) compatible CBF-CLF designs.
\end{abstract}

\section{Introduction}

In many control applications, stability alone is no longer sufficient; formal safety guarantees are equally necessary. Feedback controllers must not only stabilize the system but also enforce state constraints, a requirement that is especially critical in domains such as autonomous vehicles, industrial processes, and robotics, where safe operation is essential for real-world deployment.

Control Lyapunov functions (CLFs) provide a classical framework for stabilizing nonlinear systems, with necessary and sufficient conditions for their existence established in \cite{artstein1983stabilization} and later exploited to derive a universal control law for control-affine nonlinear systems \cite{sontag1989universal}. Similarly, control barrier functions (CBFs) \cite{wieland2007constructive,ames2016control,xu2015robustness} enforce safety by constraining the Lie derivative of a barrier function, but they alone do not guarantee stability. This gap has motivated increasing interest in safe stabilization and reach-avoid formulations that combine safety with stability or reachability, as studied in works such as \cite{romdlony2016stabilization,meng2022smooth,mestres2025conversetheoremscertificatessafety,li2024stabilization,dawson2022safe,li2023graphical,ong2019universal,mestres2022optimization,quartz2025converse}.

A widely used strategy is to couple CBFs and CLFs within an optimization-based framework. For control-affine systems, this is typically formulated as a quadratic program (QP) \cite{ames2016control,ames2019control}, which can be solved efficiently and deployed online, for instance through model predictive control (MPC) \cite{wu2019control}. The main challenge lies in guaranteeing compatibility between the two functions so that the optimization remains feasible. When compatibility is absent, stability is often relaxed and treated as a soft constraint \cite{ames2016control,li2023graphical,mestres2022optimization}. A further limitation is that such optimization-based formulations may introduce undesired dynamics such as additional equilibrium points \cite{mestres2022optimization}, and the provable region of attraction, if any, is often considerably smaller than the certified safe region.

An alternative approach is to merge a CLF and a CBF into a single control Lyapunov-barrier function (CLBF) \cite{romdlony2016stabilization}. Once such a function is available, Sontag's universal formula \cite{sontag1989universal} (or QP with guaranteed feasibility) can be applied to construct a safe stabilizing controller. While conceptually appealing, the drawback---as highlighted in \cite{braun2017existence,braun2020comment} and further discussed in \cite{meng2022smooth,mestres2025conversetheoremscertificatessafety}---is that the CLBF conditions of \cite{romdlony2016stabilization} generally fail to hold unless stringent assumptions on the safe set are satisfied. The difficulty stems from topological obstructions: a system with state constraints, whether imposed explicitly by safety requirements or implicitly by limited controllability, may not admit a continuous stabilizing feedback. We refer the reader to the expository work \cite{sontag1999stability} for a discussion on the necessity of discontinuous feedback in nonlinear stabilization.  

Converse Lyapunov theorems for safety \cite{liu2020converse} and for stability under safety constraints \cite{meng2022smooth} are useful for theoretically understanding when a single Lyapunov function can certify both stability and safety. The work in \cite{liu2020converse} provided a characterization of robust safety using Lyapunov functions. In \cite{meng2022smooth}, a general theoretical framework was established that connects Lyapunov and barrier functions through converse results, ensuring stability together with safety and accommodating reach-avoid-stay specifications. This framework has also been extended to cover hybrid systems \cite{meng2023lyapunov} and stochastic systems \cite{meng2024stochastic}. Building on \cite{meng2022smooth}, 
the recent work  \cite{mestres2025conversetheoremscertificatessafety} broadened the scope by deriving necessary conditions for the existence of CLBFs and for the compatibility of CBF-CLF pairs. Complementary to these results, the converse Lyapunov theorem developed in \cite{quartz2025converse} shows that a strictly compatible pair of control Lyapunov and control barrier functions exists if and only if there is a single smooth Lyapunov function that simultaneously certifies asymptotic stability and safety. 

Inspired by the work above, and in particular \cite{quartz2025converse}, the present paper develops a computational framework to unify strictly compatible CBFs and CLFs into a single CLBF for safe stabilization. Our first motivation is to enlarge the safe stabilization region compared with current approaches based on compatible CBF-CLF pairs \cite{li2023graphical,dai2024verification}. To this end, we propose a natural log-sum-exp (softmax) relaxation of a maximum barrier function, obtained from sublevel safety constraints, as a candidate CBF. Surprisingly, this simple choice often yields barrier functions that are formally verifiable. When verification fails, we introduce a counterexample-guided synthesis procedure that iteratively inserts half-space cuts until a verifiable barrier is obtained. We find that this significantly extends the safe stabilization region achieved by existing CBF-CLF synthesis methods \cite{li2023graphical,dai2024verification}, as it does not restrict barrier functions to a fixed template and instead allows them to adapt to the geometry of the safe set. We then search for a compatible CLF and prove theoretically that, under a strict compatibility condition, it is always possible to patch the compatible CBF-CLF pair into a single smooth CLBF. In contrast to methods that require compatibility over larger domains, our approach only imposes strict compatibility on the boundary of the safe set, making the search for compatible CLFs more tractable. Through a series of benchmark examples, we demonstrate that the proposed method outperforms alternative approaches, including sum-of-squares (SOS) synthesis of compatible CBF-CLF pairs.

\section{Preliminaries and problem formulation}

\subsection{System description}

Consider a nonlinear system in control-affine form
\begin{equation}
    \label{eq:sys}
    \dot x = f(x) + g(x)u,
\end{equation}
where $f:\,\mathbb{R}^n \to \mathbb{R}^n$ and $g:\,\mathbb{R}^n \to \mathbb{R}^{n \times m}$ with $f(0)=0$. 
We aim to design a state-feedback controller $u = \kappa(x)$, where $\kappa:\,\mathbb{R}^n \to \mathbb{R}^m$ and $\kappa(0)=0$, such that the origin remains an equilibrium point of the closed-loop system
\begin{equation}
    \label{eq:clsys}
    \dot x = f(x) + g(x)\kappa(x),
\end{equation}
and solutions of (\ref{eq:clsys}) satisfy additional properties. In this paper, we focus on asymptotic stabilization under state constraints. 

Given a control input $u:[0,\infty) \to \mathbb{R}^m$ and functions $f$ and $g$ satisfying mild conditions, for example $u$ is locally essentially bounded and $f$, $g$ are locally Lipschitz, system~(\ref{eq:sys}) admits a unique solution $\phi(t; x, u)$ on its maximal interval of existence. Here, $\phi(t; x, u)$ denotes the unique trajectory of~(\ref{eq:clsys}) starting from $x(0)=x$ under the control input $u$. We call an input signal admissible if a solution exists for it.  

\subsection{Safety constraints}

Consider multiple state constraints of the form
\begin{equation}
\label{eq:constraint}
\mathcal{C}_i = \left\{ x \in \mathbb{R}^n \mid h_i(x) \le 1 \right\}, \quad i = 1, \ldots, N,
\end{equation}
where each $h_i:,\mathbb{R}^n \to \mathbb{R}$ is continuously differentiable. We assume that the safe set is defined as the intersection of the sets above:
\begin{equation}
\label{eq:safe_set}
\mathcal{C}_{\text{safe}}=\bigcap_{i=1}^N\mathcal{C}_i.
\end{equation}
Equivalently, the safe set can be written as
\begin{equation}
\label{eq:hmax}
\mathcal{C}_{\text{safe}} = \left\{ x \in \mathbb{R}^n \mid h_{\max}(x) \le 1 \right\},
\end{equation}
where $h_{\max}(x) := \max_{1 \le i \le N} h_i(x)$.
Note that $h_{\max}$ is generally not continuously differentiable.

Safety is enforced by requiring that solutions of the closed-loop system~(\ref{eq:clsys}) remain in the set $\mathcal{C}_{\text{safe}}$ whenever they start in $\mathcal{C}_{\text{safe}}$.  
Additionally, we impose asymptotic stability. Assuming that the origin lies in the interior of $\mathcal{C}_{\text{safe}}$, we seek a feedback controller $u = \kappa(x)$ such that the origin is asymptotically stable for the closed-loop system~(\ref{eq:clsys}), and solutions starting in $\mathcal{C}_{\text{safe}}$ remain in $\mathcal{C}_{\text{safe}}$ and converge to the origin as $t \to \infty$.

This is not always possible unless the set $\mathcal{C}_{\text{safe}}$ is a controlled forward invariant set for~(\ref{eq:sys}) and is contained within the domain of null controllability of~(\ref{eq:sys}).  

\begin{defn}
A set $\mathcal{C} \subset \mathbb{R}^n$ is called controlled forward invariant for system~(\ref{eq:sys}) if for every $x \in \mathcal{C}$ there exists an admissible control input $u$ such that $\phi(t; x, u)$ is defined for all $t \ge 0$ and satisfies $\phi(t; x, u) \in \mathcal{C}$ for all $t \ge 0$.
\end{defn}

\begin{defn}
The domain of null controllability for~(\ref{eq:sys}) is defined as  
\begin{equation*}
\mathcal{D} := \left\{ x \in \mathbb{R}^n \,\middle|\, \exists\, u \text{ such that } \phi(t; x, u) \to 0 
\text{ as } t \to \infty \right\}.
\end{equation*}
\end{defn}

\subsection{Problem formulation}\label{sec:problem}

We seek to find a tight under-approximation of $\mathcal{C}_{\text{safe}}$, denoted by $\mathcal{C}$, that is represented as the 1-sublevel set of a continuously differentiable function $W:\,\Real^n\ra\Real$, i.e.,
\begin{equation}
    \label{eq:set_C}
    \mathcal{C}=\set{x\in\mathcal \Real^n \mid W(x)\le 1}
\end{equation}
Furthermore, we require that $W$ serves as a control Lyapunov function for~(\ref{eq:sys}) on $\C$; that is,  $W$ is positive definite and 
\begin{equation}
\inf_{u \in \mathbb{R}^m} [L_f W(x) + L_g W(x)u] < 0,  
\quad \forall x \in \mathcal{C} \setminus \{0\},
\end{equation}
where $L_f W = \nabla W^\top f $ and $L_g W =  \nabla W^\top g$. As a result, standard results from nonlinear control (e.g., Sontag's universal formula~\cite{sontag1989universal}) apply directly to $W$ to obtain a safe, stabilizing controller\textemdash one of the main benefits of having a single smooth Lyapunov function. These results ensure we can construct a feedback controller that is continuous everywhere except possibly at the origin. 

\begin{rem}
We take a unifying approach that employs a single Lyapunov function to certify both safety and stability, in contrast to the common approach of using a separate barrier function and a Lyapunov function~\cite{ames2016control}.  
It has been shown in the literature that, under mild assumptions, this can be done without loss of generality; see, for example, \cite{liu2020converse,meng2022smooth,mestres2025conversetheoremscertificatessafety,quartz2025converse} for characterizations of safety and stability using converse Lyapunov functions.  
In this work, we focus on demonstrating the computational feasibility and benefits of unifying barrier and Lyapunov functions. Despite being merely a control Lyapunov function, $W$ is also referred to as a control Lyapunov-barrier function (CLBF) to emphasize its dual role in simultaneously certifying stability and safety, in line with recent work on unified Lyapunov-barrier certificates \cite{romdlony2016stabilization,meng2022smooth,mestres2025conversetheoremscertificatessafety,quartz2025converse}. 
\end{rem}

\section{Softmax relaxation of the safe set}\label{sec:softmax}

Often, one of the objectives of a safe stabilizing controller is to maximize the safe domain of attraction,  
namely the set of initial conditions that converge to the origin while remaining in the safe set $\mathcal{C}_{\text{safe}}$.  
A natural choice is to design a barrier function whose $1$-sublevel set coincides with $\mathcal{C}_{\text{safe}}$.  
Indeed, $h_{\max}$ in~(\ref{eq:hmax}) provides an exact characterization of $\mathcal{C}_{\text{safe}}$ via its $1$-sublevel set.  
A potential drawback of using $h_{\max}$ directly, however, is that it is generally not differentiable.

\subsection{Softmax relaxation of the safe set}

We approximate $h_{\max}$ by the smooth log-sum-exp with temperature $\tau>0$:
\begin{equation}
\label{eq:softmax}
h_{\mathrm{sm}}(x;\tau)\;:=\;\frac{1}{\tau}\,\log\!\Big(\sum_{i=1}^N e^{\tau\,h_i(x)}\Big).
\end{equation}
We have the following well-known property of $h_{\mathrm{sm}}$. The proof is elementary, but we include it for completeness.  
\begin{proposition}\label{prop:softmax}
For all $x \in \mathbb{R}^n$,
\[
h_{\max}(x) \;\le\; h_{\mathrm{sm}}(x;\tau) \;\le\; h_{\max}(x) + \frac{\log N}{\tau}.
\]
\end{proposition}
\begin{proof}
Fix $x\in\Real^n$. Let $a_i := h_i(x)$ for $i=1,\dots,N$ and $a_{\max} := \max_i a_i$.  
Then
\[
\sum_{i=1}^N e^{\tau a_i} 
\;\ge\; e^{\tau a_{\max}},
\qquad
\sum_{i=1}^N e^{\tau a_i}
\;\le\; N e^{\tau a_{\max}},
\]
since $e^{\tau a_i} \le e^{\tau a_{\max}}$ for all $i$.  
Taking logarithms and dividing by $\tau>0$ gives
\[
a_{\max} 
\;\le\; \frac{1}{\tau}\log\!\Bigl(\sum_{i=1}^N e^{\tau a_i}\Bigr)
\;\le\; a_{\max} + \frac{\log N}{\tau},
\]
which is precisely the desired inequality.
\end{proof}
An immediate consequence of Proposition \ref{prop:softmax} is that 
\begin{equation}
\{h_{\mathrm{sm}}\le 1\}\subseteq\{h_{\max}\le 1\}=\mathcal{C}_{\text{safe}}.
\end{equation}
In other words, the 1-sublevel set of $h_{\mathrm{sm}}$ gives a guaranteed under-approximation of the safe set, and, as $\tau\to\infty$, $h_{\mathrm{sm}}\to h_{\max}$ uniformly on compact sets.

\subsection{Control barrier condition}

Consider a continuously differentiable function $h: \mathbb{R}^n \rightarrow \mathbb{R}$. Denote 
\begin{align} 
\mathcal{C} & = \left\{ x \in \mathbb{R}^n \mid h(x) \le 1 \right\}, \label{safe-set}\\
\partial \mathcal{C} & = \left\{ x \in \mathbb{R}^n \mid h(x) = 1 \right\}.
\end{align}
We rely on the following strict barrier condition on $h$ to ensure controlled invariance of the set $\mathcal C$.
\begin{defn}
We say that $h$ is a {strict control barrier function} on $\mathcal C$ if, for every $x\in \partial \mathcal C$, there exists $u\in\Real^m$ such that
\begin{equation}\label{eq:cbf}
    \inf_{u \in \mathbb{R}^m} [L_f h(x) + L_g h(x) u] < 0.
\end{equation}
\end{defn}

\begin{example}\label{ex:2d_toy}
We consider the following toy example taken from~\cite{dai2024verification} to illustrate the
construction of a near-maximal smooth barrier function and its verification.

\paragraph{System dynamics}
Let $x = (x_1,x_2)^\top$ and consider the control-affine system (\ref{eq:sys}) with
\[
    f(x) = \begin{bmatrix} 0 \\ -\sin x_1 \end{bmatrix},
    \quad
    g(x) = \begin{bmatrix} 1 \\ -1 \end{bmatrix},
\]
and domain $[-\pi,\pi] \times [-3,4]$.

\paragraph{Hard constraints}
The safe set is defined by 
\[
    h_1(x) := -\sin x_1 - \cos x_1 - x_2 \;\le 1,
\]
together with box constraints encoding the domain:
\[
    h_{2j}(x) := 1 + x_j - \bar x_j \le 1, \qquad
    h_{2j+1}(x) := 1 - x_j + \underline x_j \le 1,
\]
for $j=1,2,$ where $\underline x_j$ and $\bar x_j$ denote the lower and upper bounds of the $j$th coordinate.

\paragraph{Verified smooth barrier}
Let $h_1,\dots,h_N$ ($N=5$) be all constraints above. We construct a
smooth \emph{log-sum-exp} approximation $h_{\mathrm{sm}}(x;\tau)$ as defined in (
\ref{eq:softmax}). For a wide range of $\tau$ values, we can formally verify (details to provide in Section \ref{sec:verification}) the strict barrier condition (\ref{eq:cbf}). Fig. \ref{fig:softmax-barrier} depicts formally verified $h_{\mathrm{sm}}(x;\tau)$ for $\tau=1.5$ and $\tau=4.5$. It clearly shows that increasing $\tau$ leads to a closer approximation of the exact safe set boundary. As a comparison, we also depict a state-of-the-art approach that synthesizes compatible Lyapunov and barrier functions using sum-of-squares (SOS) as a comparison. 
Of course, at the moment, we have only discussed the barrier condition. With the results to be presented in Section~\ref{sec:patched_clbf}, we will be able to produce a formally verified single smooth Lyapunov-barrier function that achieves the exact same safe stabilization region enclosed by the softmax barrier.

\begin{figure}[h!]
    \centering
    \includegraphics[width=\linewidth]{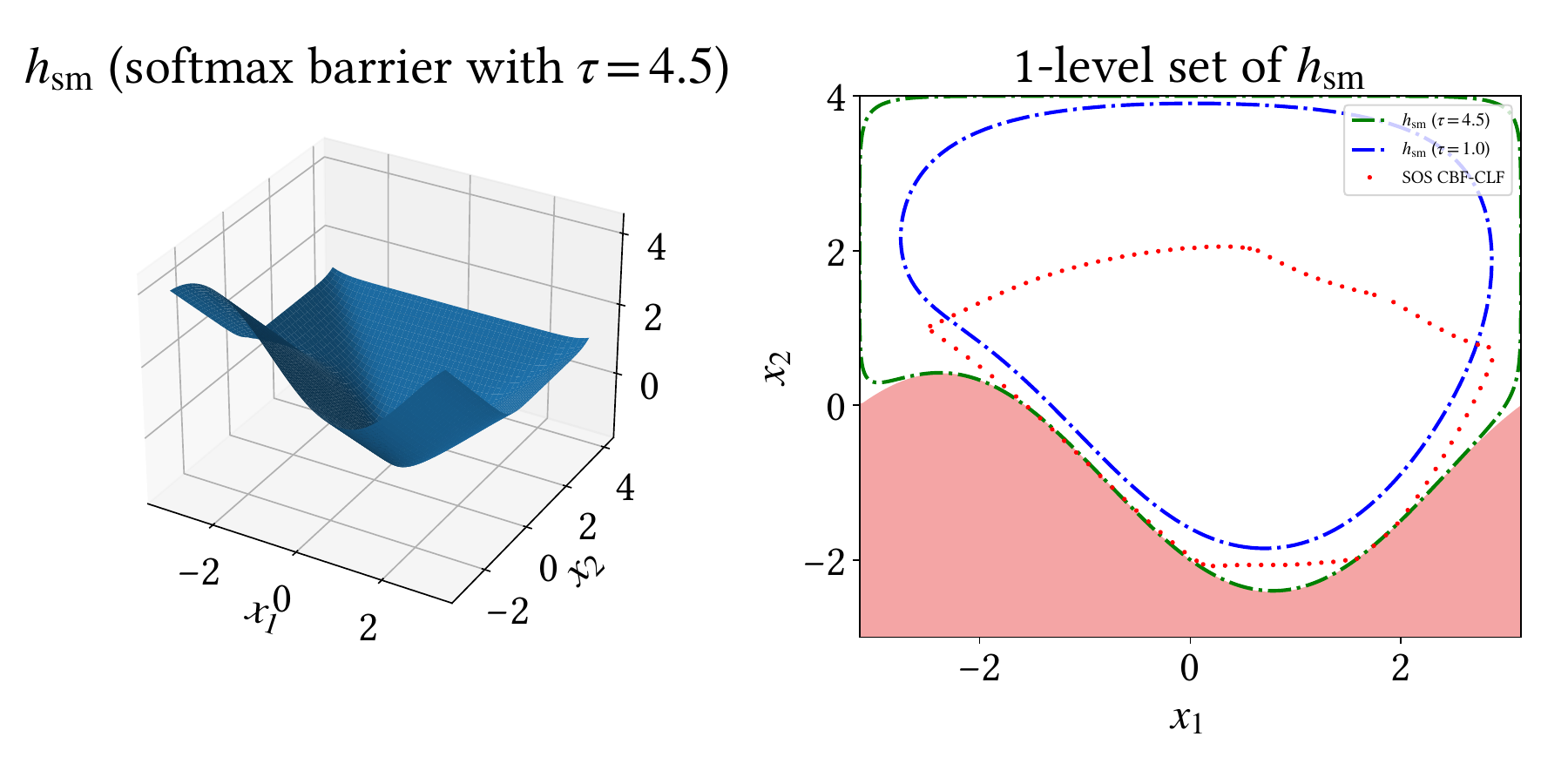}
    \caption{Smooth barrier function $h_{\mathrm{sm}}(x)$ (left) and its $1$-level set (right). As a comparison, the best verified safe controllable region obtained using SOS CBF-CLF is shown in dotted red. It can be seen that the softmax relaxation barrier outperforms the SOS approach \cite{dai2024verification}. Data for comparisons in this and subsequent examples were extracted from published figures using the WebPlotDigitizer tool~\cite{marin2017webplotdigitizer}.}
    \label{fig:softmax-barrier}
\end{figure}

\end{example}

\subsection{Counterexample-guided refinement}

While in many cases, as we shall see in Section~\ref{sec:examples}, the
softmax relaxation barriers readily satisfy the strict barrier condition,
there are scenarios in which this is not the case. 
We present a counterexample-guided refinement procedure that iteratively
adds new half-space constraints to the softmax approximation until the 
barrier condition becomes verifiable.

\paragraph{Construction of half-space cuts}
Suppose that $h_{\mathrm{sm}}(x;\tau)$ fails the barrier condition at some
counterexample $x^\ast\in\partial \mathcal C$ provided by the verifier.
Let 
\[
n := \frac{\nabla h_{\mathrm{sm}}(x^\ast;\tau)}{\|\nabla h_{\mathrm{sm}}(x^\ast;\tau)\|}
\]
denote the outward unit normal of the current softmax level set at $x^\ast$.
To exclude this point, we generate a new half-space
\begin{equation}
\label{eq:new-halfspace}
h_{\mathrm{new}}(x) := n_{\mathrm{new}}^\top x - b_{\mathrm{new}} + 1,
\qquad
h_{\mathrm{new}}(x) \le 1,
\end{equation}
where $n_{\mathrm{new}}$ is a slight rotation of $n$:
\[
n_{\mathrm{new}}
:= n\cos\theta + r\sin\theta,
\qquad
r \perp n,\ \|r\|=1,
\]
with a small angle $\theta > 0$. 
The plane is shifted inward by a small margin $\varepsilon>0$,
$
b_{\mathrm{new}} := n_{\mathrm{new}}^\top x^\ast - \varepsilon,
$ 
so that $x^\ast$ lies strictly outside the updated safe set.
The refined constraint set is then
\[
\{h_1(x)\le 1,\dots,h_N(x)\le 1,h_{\mathrm{new}}(x)\le 1\},
\]
and a new softmax $h_{\mathrm{sm}}(x;\tau)$ is built from this set.

\paragraph{Algorithm}
The procedure repeats this step until either the barrier condition is 
verified or a maximum number of cuts is reached.
This is summarized in Algorithm~\ref{alg:cegis-softmax}.

\begin{algorithm}[h!]
\caption{Counterexample-Guided Softmax Barrier Refinement}
\label{alg:cegis-softmax}
\KwIn{Initial constraints $\{h_i(x)\le 1\}_{i=1}^N$, temperature $\tau>0$}
Build $h_{\mathrm{sm}}(x;\tau)$ via (\ref{eq:softmax})\;
\While{Verifier finds counterexample $x^\ast$ and iteration count $<k_{\max}$}{
    Compute normal $n \gets \nabla h_{\mathrm{sm}}(x^\ast;\tau)/\|\nabla h_{\mathrm{sm}}(x^\ast;\tau)\|$\;
    Choose small angle $\theta>0$ and orthogonal direction $r \perp n$\;
    Set $n_{\mathrm{new}}\gets n\cos\theta + r\sin\theta$\;
    Set $b_{\mathrm{new}}\gets n_{\mathrm{new}}^\top x^\ast - \varepsilon$\;
    Add new half-space constraint $h_{\mathrm{new}}(x) := n_{\mathrm{new}}^\top x - b_{\mathrm{new}} + 1$\;
    Rebuild $h_{\mathrm{sm}}(x;\tau)$ with augmented constraint set\;
}
\KwOut{Refined $h_{\mathrm{sm}}(x;\tau)$ satisfying barrier condition}
\end{algorithm}

\begin{example}\label{ex:2d_nonlinear}
We consider the nonlinear control-affine system (\ref{eq:sys}), taken from \cite{dai2024verificationv1}, with 
\[
f(x)=\begin{bmatrix} 0 \\ -x_1 + x_1^3/6 \end{bmatrix},
\quad
g(x)=\begin{bmatrix} 1 \\ -1 \end{bmatrix},
\]
and domain $[-4.5,4.5]^2$. The unsafe set is given by the single half-space
$h_1(x):= -2 - x_1 - x_2 \le 1$, together with the box constraints
defining the domain. We construct a softmax relaxation $h_{\mathrm{sm}}(x;\tau)$ with $\tau=3.0$. Direct verification of the control barrier condition on $h_{\mathrm{sm}}=1$
fails, but applying the counterexample-guided refinement procedure from
Algorithm~\ref{alg:cegis-softmax} adds a sequence of half-space cuts.
Upon termination, the refined softmax barrier is formally verified, and its
$1$-sublevel set provides a guaranteed controlled invariant subset of the
original safe set. Fig. \ref{fig:ex_refined_barrier} depicts the verified refined softmax barrier for this example.    
\begin{figure}[h!]
    \centering
    \includegraphics[width=\linewidth]{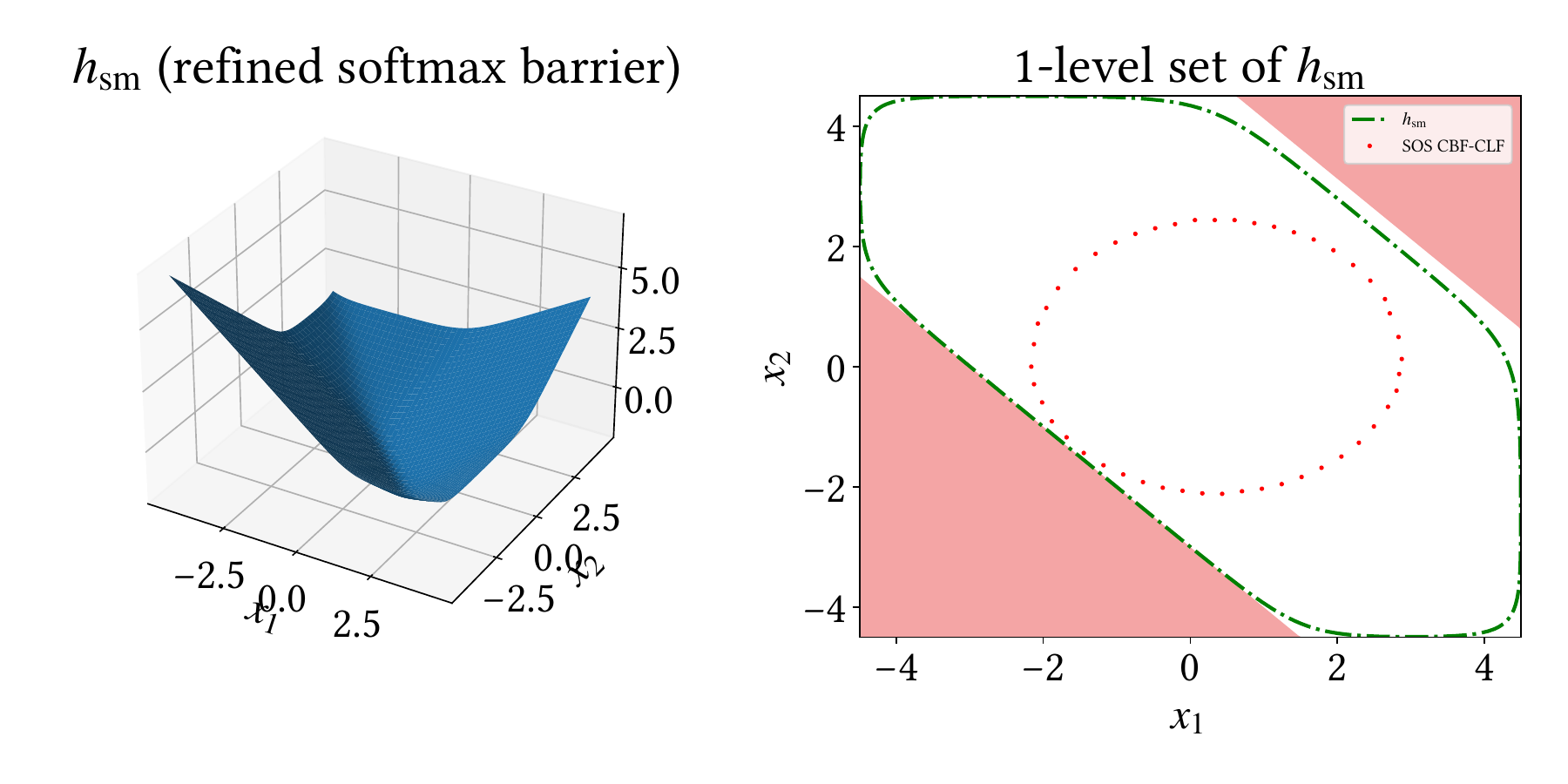}
    \caption{Smooth barrier function $h_{\mathrm{sm}}(x)$ (left) and its $1$-level set (right) obtained by 
softmax relaxation and counterexample-guided refinement (Algorithm \ref{alg:cegis-softmax}. As a comparison, the best verified safe controllable region obtained using SOS CBF-CLF is shown in dotted red. The softmax relaxation barrier outperforms the SOS approach \cite{dai2024verificationv1}.}
    \label{fig:ex_refined_barrier}
\end{figure}
\end{example}

\section{Strict compatibility implies smooth patching}
\label{sec:patched_clbf}

In the previous section, we demonstrated that the softmax relaxation can yield a significantly less conservative barrier for safety. To ensure safe stabilization, however, we must combine it with a Lyapunov function for stabilization, posed as the main problem in Section~\ref{sec:problem}. We address this issue in the present section. The theoretical foundation is a converse Lyapunov theorem for joint safety and stability recently proved in~\cite{quartz2025converse}. Our focus here is on efficient and verifiable computation. The proposed approach uses smooth patching with formal guarantees and relies on satisfiability modulo theories (SMT) verification and an analytical formula to patch strictly compatible control barrier and Lyapunov functions. We present this as the main result of the paper.

\subsection{Strictly compatible CBF and CLF}

Consider a candidate CBF $h:\,\Real^n\ra\Real$ and a candidate CLF  
$V:\,\Real^n\ra\Real$. We assume that both $h$ and $V$ are continuously differentiable, $V$ is positive definite, and the set  $\mathcal{C}$ defined by (\ref{safe-set}) contains the origin in its interior. 

\begin{defn}\label{def:compatibility}
    We say that $h$ and $V$ are strictly compatible if the following conditions hold: 
    \begin{enumerate}
        \item For every $x \in \mathcal{C} \setminus \{0\}$, there exists $u\in\Real^m$ such that  
    \begin{equation}\label{eq:CLFV}
        L_f V(x) + L_g V(x) u < 0.
    \end{equation}
        \item For every $x \in \partial \C$, there exists $u\in\Real^m$ such that 
    \begin{equation}\label{eq:CLFBoundary}
        L_f V(x) + L_g V(x) u < 0,
    \end{equation}
    and
    \begin{equation}\label{eq:CBF}
        L_f h(x)+L_g h(x) u < 0.
    \end{equation}
    \end{enumerate}
\end{defn}

\begin{rem}
We impose the strict form of inequality~\eqref{eq:CBF} to ensure robust invariance of $\partial\mathcal{C}$. This is essential for it to be verifiable by $\delta$-complete SMT solvers such as dReal \cite{gao2013dreal}. We also relax the compatibility conditions, see, e.g., \cite{mestres2022optimization,ong2019universal,dai2024verification}, to hold only on the boundary of the set \(\C\).
\end{rem}

\subsection{Smooth patching with guarantees}

We establish the following result on patching a strictly compatible CBF and CLF pair to form a single smooth Lyapunov-barrier function (CLBF). 

\begin{thm}\label{thm:patch}
    Let $h$ and $V$ be strictly compatible as defined in Definition \ref{def:compatibility}. Suppose that the set $\C$ is compact. Then there exists a continuously differentiable function $W:\,\Real^n\ra\Real$ with the following properties:
    \begin{enumerate}
        \item $\mathcal{C}=\set{x\in\mathcal \Real^n \mid W(x)\le 1}$.
        \item For every $x \in \mathcal{C} \setminus \{0\}$, there exists $u\in\Real^m$ such that  
    \begin{equation}\label{eq:CLFW}
        L_f W(x) + L_g W(x) u < 0.
    \end{equation}
    \end{enumerate}
\end{thm}

\begin{proof}
Fix $\varepsilon \in (0,1)$ and define the inner boundary band 
\[
\partial\mathcal C_{\varepsilon}^{-} := \{x : 1-\varepsilon \le h(x) \le 1\}.
\]
By strict compatibility and compactness of $\C$, we can choose $\eps>0$ such that, for every $x \in \partial\mathcal C_{\varepsilon}^{-}$, there exists $u$ such that both
$L_f V(x) + L_g V(x) u < 0$ and $L_f h(x) + L_g h(x) u < 0$. We also assume that $\eps>0$ is sufficiently small such that $0\not\in \partial\mathcal C_{\varepsilon}^{-}$. 

Since $\mathcal C$ is compact, we can choose $\alpha > 0$ such that $V_2 := \alpha V$ satisfies
\[
\quad 
V_2(x) \le 1 \ \forall x \in \mathcal C,
\qquad
V_2(x) \le h(x) \ \forall x \in \partial\mathcal C_{\varepsilon}^{-}.
\]
Indeed, we can simply choose 
$$
\alpha = \frac{1-\eps}{\max_{x\in \C} V(x)}.
$$

Define a $C^1$ bump function supported on the band:
\[
b(x)=
\begin{cases}
\exp\!\Bigl(-\dfrac{1}{\varepsilon^2 - (h(x)-1)^2} + \dfrac{1}{\varepsilon^2}\Bigr), & 
 1-\varepsilon < h(x) < 1,\\[4pt]
1, & h(x) \ge 1,\\
0, & h(x) \le 1-\varepsilon.
\end{cases}
\]  
Now define the patched function
\begin{equation}\label{eq:patched_W}
W(x) := (1-b(x)) V_2(x) + b(x) h(x),
\end{equation}
which is $C^1$ by construction, with gradient
\begin{equation}\label{eq:grad_W}
\nabla W = b \nabla h + (1-b)\nabla V_2 + (h - V_2)\nabla b.    
\end{equation}

We verify the following properties of $W$. 

\emph{Property (1): Exact safe set.}  
If $h > 1$, then $b=1$ and $W=h > 1$. This implies 
$\{W \le 1\} \subset \{h \le 1\}$. 

If $h \le 1-\varepsilon$, then $b=0$ and $W=V_2 \le 1$.  
If $1-\varepsilon<h<1$, then $0<b<1$ and
$W=(1-b)V_2 + b h \le 1$ since $V_2 \le 1$ on $\mathcal C$ and $h \le 1$ in the band. These together imply 
$\{h \le 1\} \subset \{W \le 1\}$. 

Thus $\{W \le 1\} = \{h \le 1\} = \mathcal C$.

\emph{Property (2): Strict CLF condition.} Following (\ref{eq:grad_W}), we further compute
\[
\nabla b(x) =
\begin{cases}
p(x)\,\nabla h(x), & \quad \text{if } 1-\varepsilon < h(x) < 1, \\[6pt]
0, & \quad \text{if }  h(x) \ge 1 \text{ or } h(x) \le 1-\varepsilon,
\end{cases}
\]
where
$$
p(x) =  \exp\!\Bigl(-\dfrac{1}{\varepsilon^2 - (h-1)^2} + \dfrac{1}{\varepsilon^2}\Bigr)
\frac{2(1-h)}{\bigl(\varepsilon^2-(h-1)^2\bigr)^2}.
$$
Take $x \in \mathcal C \setminus \{0\}$.
\begin{itemize}
\item If $h(x) \le 1-\varepsilon$, then $b=0$, $\nabla b=0$, 
and $\nabla W = \nabla V_2 = \alpha \nabla V$. The CLF condition (\ref{eq:CLFW}) holds for $W$ because 
$V$ satisfies the CLF condition (\ref{eq:CLFV}).
\item If $1-\varepsilon < h(x) \le 1$, then
\[
\nabla W = \bigl(b+(h - V_2)p\bigr)\nabla h + (1-b)\nabla V_2,
\]
where $p\ge 0$ and $h - V_2 \ge 0$. By strict compatibility of $h$ and $V$, for each $x$,  
there exists $u$ such that $\nabla h^\top (f+gu)<0$ and $\nabla V_2^\top (f+gu)<0$ simultaneously. This implies $\nabla W^\top (f+gu)< 0$.
\end{itemize}
Therefore, for every $x \in \mathcal C \setminus \{0\}$, there exists $u$ such that
$L_f W(x) + L_g W(x) u < 0$.
\end{proof}

The following corollary is a direct application of Theorem~\ref{thm:patch} 
and is a well-established result in nonlinear control.

\begin{cor}\label{cor:controller}
    Under the assumptions of Theorem \ref{thm:patch}, there exists a feedback controller $\kappa:\,\Real^n\ra\Real$ such that the origin is asymptotically stable for the closed-loop system~(\ref{eq:clsys}), and solutions starting in $\mathcal{C}$ remain in $\mathcal{C}$ and converge to the origin as $t \to \infty$. Furthermore, $\kappa$ can be chosen to be smooth on $\Real^n\setminus\set{0}$. 
\end{cor}

\begin{rem}
If $V$ satisfies the small control property~\cite{sontag1989universal}, 
then so does $W$ constructed in~\eqref{eq:patched_W}. 
Consequently, the feedback law $\kappa$ in the corollary is continuous at the origin. 
Moreover, $\kappa$ can be explicitly constructed using Sontag's universal 
stabilization formula~\cite{sontag1989universal}.
\end{rem}

\section{SMT verification}\label{sec:verification}

In this section, we discuss formal verification of the strict barrier condition~\eqref{eq:cbf} introduced in 
Section~\ref{sec:softmax} and the strict compatibility between control barrier and Lyapunov 
conditions defined in Section~\ref{sec:patched_clbf}. Once these conditions are verified, 
the function $W$ constructed in Theorem~\ref{thm:patch} is a continuously differentiable 
CLBF that can be used for safe stabilization. 

\subsection{Verification of strict CBF}

The strict CBF condition (\ref{eq:cbf}) is equivalent to 
\begin{equation}
    \label{eq:cbf-verification}
    (L_g h(x) = 0 \wedge h(x) = 1) \Longrightarrow L_f h(x)<0.
\end{equation}
The barrier function $h$ constructed using the softmax relaxation (\ref{eq:softmax}) involves transcendental functions $\log$ and $\exp$. Over a compact domain, condition (\ref{eq:cbf-verification}) can be readily verified by a $\delta$-complete SMT solver such as dReal \cite{gao2013dreal}. 

\subsection{Verification of CBF-CLF compatibility}

The strict compatibility conditions proposed in Definition \ref{def:compatibility} include a CLF condition for $V$ on $\C$, which is equivalent to 
\begin{equation}
    \label{eq:clf-verification}
    (L_g V(x) = 0 \wedge h(x) \le 1) \Longrightarrow L_f V(x)<0.
\end{equation}
Similar to~\eqref{eq:cbf-verification}, this can be readily verified using an SMT solver such as dReal~\cite{gao2013dreal}. The synthesis of CLFs using SMT solvers and neural networks has been discussed in detail in~\cite{liu2025formally}. We refer the reader to~\cite{liu2025formally} for a detailed discussion of the formal verification of the CLF condition~\eqref{eq:clf-verification}, in particular the treatment of difficulties near the origin when using dReal, which is only $\delta$-complete and may return spurious counterexamples in a neighborhood of the origin.

Furthermore, Definition \ref{def:compatibility} includes a strict compatibility condition \eqref{eq:CLFBoundary}--\eqref{eq:CBF} between $V$ and $h$ on the boundary of $\C$. Because this condition has the quantifier structure $\forall x \, \exists u$, it cannot be readily handled by dReal, which supports only quantifier-free or purely universal formulas. We use the following variant of Farkas' lemma \cite[Proposition 6.4.3, p.~90]{matouvsek2007understanding} to addresses this issue.

\begin{lemma}[Farkas' Lemma]\label{lem:farkas}
Let $A \in \Real^{n \times m}$ and $b \in \Real^n$.  
The system $A u \le b$, $u \in \Real^m$, has a solution if and only if every nonnegative 
$y \in \Real^n$ with $y^\top A = 0^\top \in \Real^{1 \times m}$ also satisfies $y^\top b \ge 0$.
\end{lemma}

We reformulate Farkas' Lemma in a strict inequality form and add a normalization to the vector $y$.

\begin{lemma}[Farkas' Lemma Reformulation]\label{lem:strictfarkas}
Let $A \in \Real^{n \times m}$ and $b \in \Real^n$.  
The system $A u < b$, $u \in \Real^m$, has a solution if and only if every nonnegative 
$\lambda \in \Real^n$ with $\lambda^\top A = 0^\top \in \Real^{1 \times m}$ and  
$\sum_{i=1}^n \lambda_i=1$ also satisfies $\lambda^\top b > 0$.
\end{lemma}

\begin{proof}
\emph{($\Rightarrow$)} We have 
\[
\lambda^\top b\;=\;\lambda^\top(b-Au)+\underbrace{\lambda^\top A u}_{=\,0}\;=\;\sum_{i=1}^n \lambda_i\,(b_i-a_i^\top u)\;>\;0,
\]
because each $b_i-a_i^\top u>0$ and $\lambda_i\ge 0$ with $\sum_{i=1}^n \lambda_i=1$. 

\emph{($\Leftarrow$)} Suppose no $u$ satisfies $Au<b$. Then for every $u$ there exists an index $i$ such that  
$a_i^\top u \ge b_i$. Equivalently, the convex sets  
\[
\{ (Au-b) \in \Real^n : u \in \Real^m \}, 
\qquad 
\Real^n_{<0}
\]
are disjoint. By the separating hyperplane theorem \cite[Section 2.5.1]{boyd2004convex}, there exists a nonzero vector $\tilde\lambda \in \Real^n$ such that  
\[
\tilde\lambda^\top(Au-b) \;\ge\; 0\quad \forall u\in\Real^m,
\qquad
\tilde\lambda^\top c \;\le\; 0\quad \forall c \in \Real^n_{<0}.
\]
The second condition implies $\tilde\lambda \ge 0$. The first condition implies $\tilde\lambda^\top A=0$ (otherwise one can choose $u$ in the direction 
of $-(\tilde\lambda^\top A)$ to violate the inequality), and then $\tilde\lambda^\top b \le 0$. Set  
\[
\lambda = \frac{\tilde\lambda}{\sum_{i=1}^n \tilde\lambda_i},
\]
so that $\lambda \ge 0$, $\sum_i \lambda_i = 1$, $\lambda^\top A=0$, and $\lambda^\top b \le 0$. This contradicts the assumption that all such $\lambda$ satisfy $\lambda^\top b > 0$. Therefore, $\exists u$ with $Au<b$. 
\end{proof}

We now state the conditions that we can verify for strict compatibility. 

\begin{lemma}\label{lem:farkas-clbf}
The strict compatibility condition \eqref{eq:CLFBoundary}--\eqref{eq:CBF} is equivalent to
\begin{equation}
    \label{eq:farkas}
\begin{aligned}
& ((h(x)=1)  \\
& \wedge (\lambda_1 \ge 0 \wedge \lambda_2 \ge 0 \wedge \lambda_1 + \lambda_2=1) \\
& \wedge (\lambda_1 L_g V(x)+\lambda_2 L_g h(x)=0))\\ 
& \Longrightarrow \lambda_1 L_f V(x) + \lambda_2 L_f h(x) <0.
\end{aligned}
\end{equation}
\end{lemma}

\begin{proof}
Write
\[
A(x)\;=\;\begin{bmatrix}L_g V(x)\\[2pt] L_g h(x)\end{bmatrix}\in\Real^{2\times m},\;
b(x)\;=\;\begin{bmatrix}-L_f V(x)\\[2pt] -L_f h(x)\end{bmatrix}\in\Real^{2}.
\]
Fix any $x\in \partial \C$ and write $\lambda=(\lambda_1,\lambda_2)^\top$. Then the strict compatibility is $\exists u:\ A(x)u<b(x)$ componentwise, and \eqref{eq:farkas} is
\[
\forall\,\lambda\in\Real_{\ge 0}^2,\ \lambda_1+\lambda_2=1,\ \lambda^\top A(x)=0^\top\ :\ \lambda^\top b(x)<0.
\]
Their equivalence is precisely Lemma \ref{lem:strictfarkas}.
\end{proof}

\begin{rem}
We consider a strict inequality formulation and add a simplex normalization to $\lambda$ so that (\ref{eq:farkas}) can be handled by $\delta$-complete SMT solvers such as dReal \cite{gao2013dreal}, which are effective essentially only for strict inequalities over compact domains.
\end{rem}

\begin{rem}
We verify (\ref{eq:cbf-verification}), (\ref{eq:clf-verification}), and (\ref{eq:farkas}) sequentially.  
To enable guaranteed patching of CBF and CLF by Theorem \ref{thm:patch} using the explicit formula (\ref{eq:patched_W}), we also do a bisection search to compute $\eps\in(0,1)$ such that 
\begin{equation}
    \label{eq:farkas-band}
\begin{aligned}
& ((1-\eps\le h(x)\le 1)  \\
& \wedge (\lambda_1 \ge 0 \wedge \lambda_2 \ge 0 \wedge \lambda_1 + \lambda_2=1) \\
& \wedge (\lambda_1 L_g V(x)+\lambda_2 L_g h(x)=0))\\ 
& \Longrightarrow \lambda_1 L_f V(x) + \lambda_2 L_f h(x) <0.
\end{aligned}
\end{equation}
is verified. This implies strict compatibility of $h$ and $V$ on the band 
$$
\partial\mathcal C_{\varepsilon}^{-} := \{x : 1-\varepsilon \le h(x) \le 1\}.
$$
defined in the proof of Theorem \ref{thm:patch}. Upon successful verification of these inequalities, the single Lyapunov function $W$ defined by (\ref{eq:patched_W}) provides both a certificate and a means for computing provably safe, stabilizing controllers. In the next section, we demonstrate that this approach is effective and can provide less conservative safe stabilization regions, compared to alternative approaches. 
\end{rem}

\section{Numerical examples}\label{sec:examples}

We illustrate the effectiveness of the proposed approach with several nonlinear control examples. All computations were performed on a 2020 MacBook Pro with a 2 GHz quad-core Intel Core i5 processor and solved within the toolbox LyZNet \cite{liu2024tool}, supported by dReal \cite{gao2013dreal} as the verification engine. The code for all examples presented here, as well as additional examples, is available at 
\url{https://git.uwaterloo.ca/hybrid-systems-lab/lyznet} under \texttt{examples/patch-clbf}.

\begin{example}\label{ex:2d_toy_verified}
We first revisit Example~\ref{ex:2d_toy}. Upon formally verifying the softmax barrier $h=h_{\mathrm{sm}}$, defined in~(\ref{eq:softmax}) with $\tau=4.5$, as a strict barrier function using (\ref{eq:cbf-verification}), we also formally verify its strict compatibility with a quadratic CLF. The strict barrier condition was certified in less than $0.01$s, and strict compatibility in $0.42$s. The latter involved a bisection procedure to determine the largest $\varepsilon$ such that strict compatibility~(\ref{eq:farkas-band}) holds on the band $\set{x : 1-\varepsilon \le h(x) \le 1}$ (capped at $\varepsilon=0.5$). The CLF we verify is the quadratic function $V(x)=x_1^2+2x_2^2$. Figure~\ref{fig:2d_toy_verified} depicts the patched CLBF and the phase portrait of the closed-loop system under Sontag's controller \cite{sontag1989universal}. Figure~\ref{fig:2d_toy_verified_trajectories} shows 50 simulated trajectories from the set $\C=\set{h\le 1}$, along with the evaluation of $h$ to demonstrate safety. The barrier function $h$ itself cannot serve as a Lyapunov function on $\mathcal{C}$, not only because it takes negative values, but also because it has a stationary point at $x = (0.1294085, 0.94176161)^\top$ (numerically confirmed via gradient descent). Patching it with a quadratic CLF yields a provable CLBF, as demonstrated here.
\end{example}

\begin{figure}[h!]
    \centering
    \includegraphics[width=\linewidth]{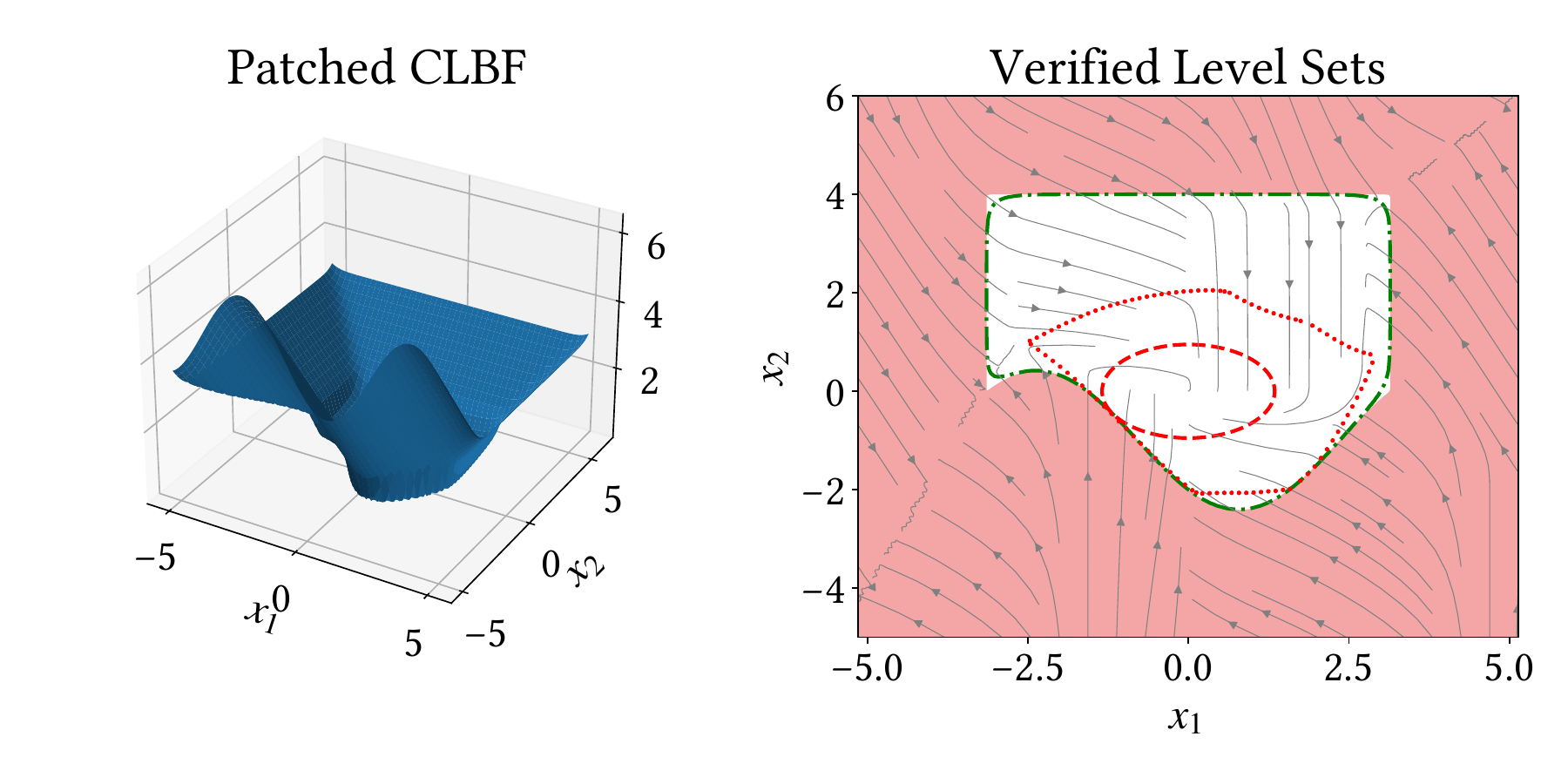}
    \caption{A smooth Lyapunov-barrier function patched from a strictly compatible CBF-CLF pair for Example \ref{ex:2d_toy_verified}. The green dash-dotted line represents the formally verified safe stabilization region, relative to the unsafe region (shaded in light red). For comparison, the best verified safe controllable region obtained using compatible SOS CBF-CLF~\cite{dai2024verification} is shown in dotted red, while the dashed red curve depicts the largest safe stabilization region certified by a quadratic CLF.}
    \label{fig:2d_toy_verified}
\end{figure}

\begin{figure}[h!]
    \centering
    \includegraphics[width=0.47\linewidth]{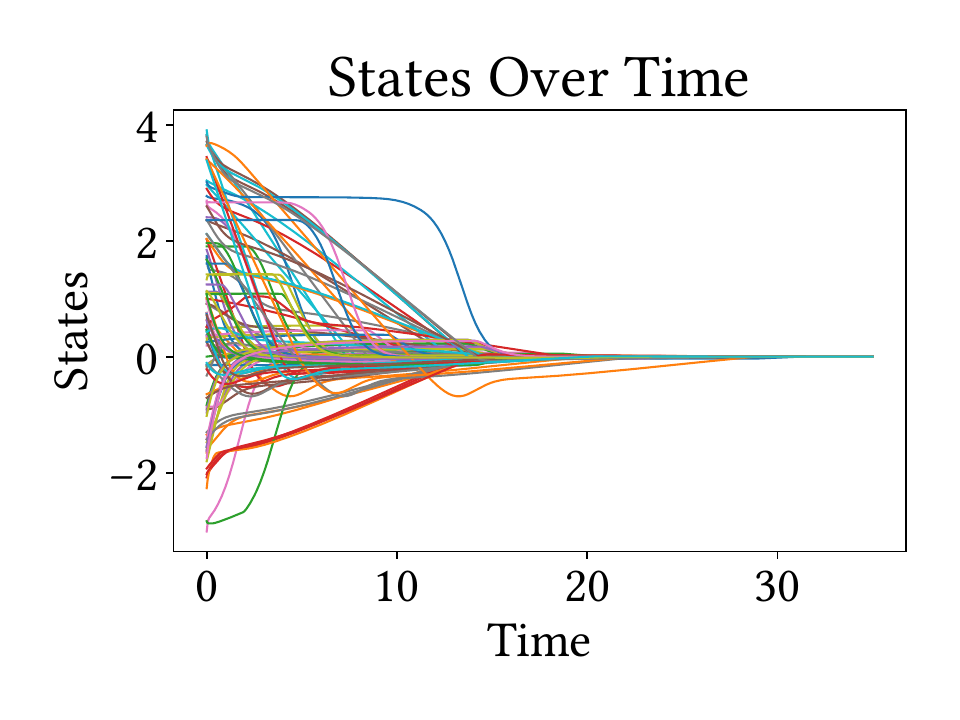}
    \includegraphics[width=0.47\linewidth]{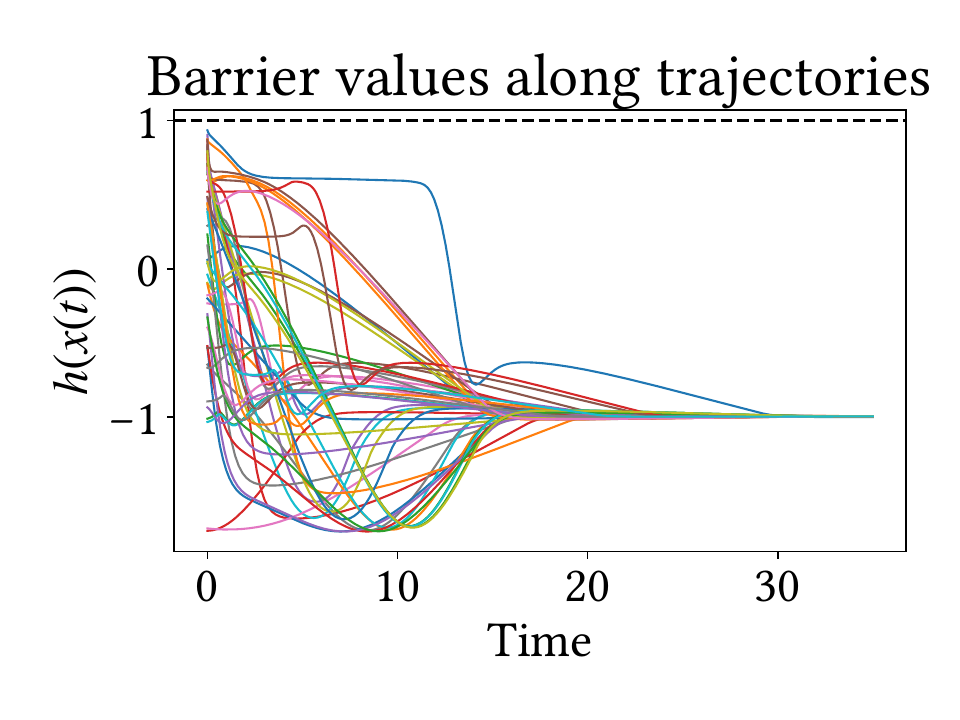}
    \caption{Simulated trajectories and corresponding barrier function values for Example \ref{ex:2d_toy_verified}. All trajectories converge, and the barrier function values remain below the safe threshold of 1.}
    \label{fig:2d_toy_verified_trajectories}
\end{figure}

\begin{example}\label{ex:2d_nonlinear_verified}
    Revisit Example \ref{ex:2d_nonlinear}. We are able to formally verify strict compatibility of the CBF $h$ depicted in Fig. \ref{fig:ex_refined_barrier} with the simple CLF $V(x) = x_1^2 + x_2^2$. The patched CLBF using (\ref{eq:patched_W}) and Theorem \ref{thm:patch} is shown in Fig. \ref{fig:2d_nonlinear_verified}, along with the phase portrait of the closed-loop system. 
\end{example}

\begin{figure}[h!]
    \centering
    \includegraphics[width=\linewidth]{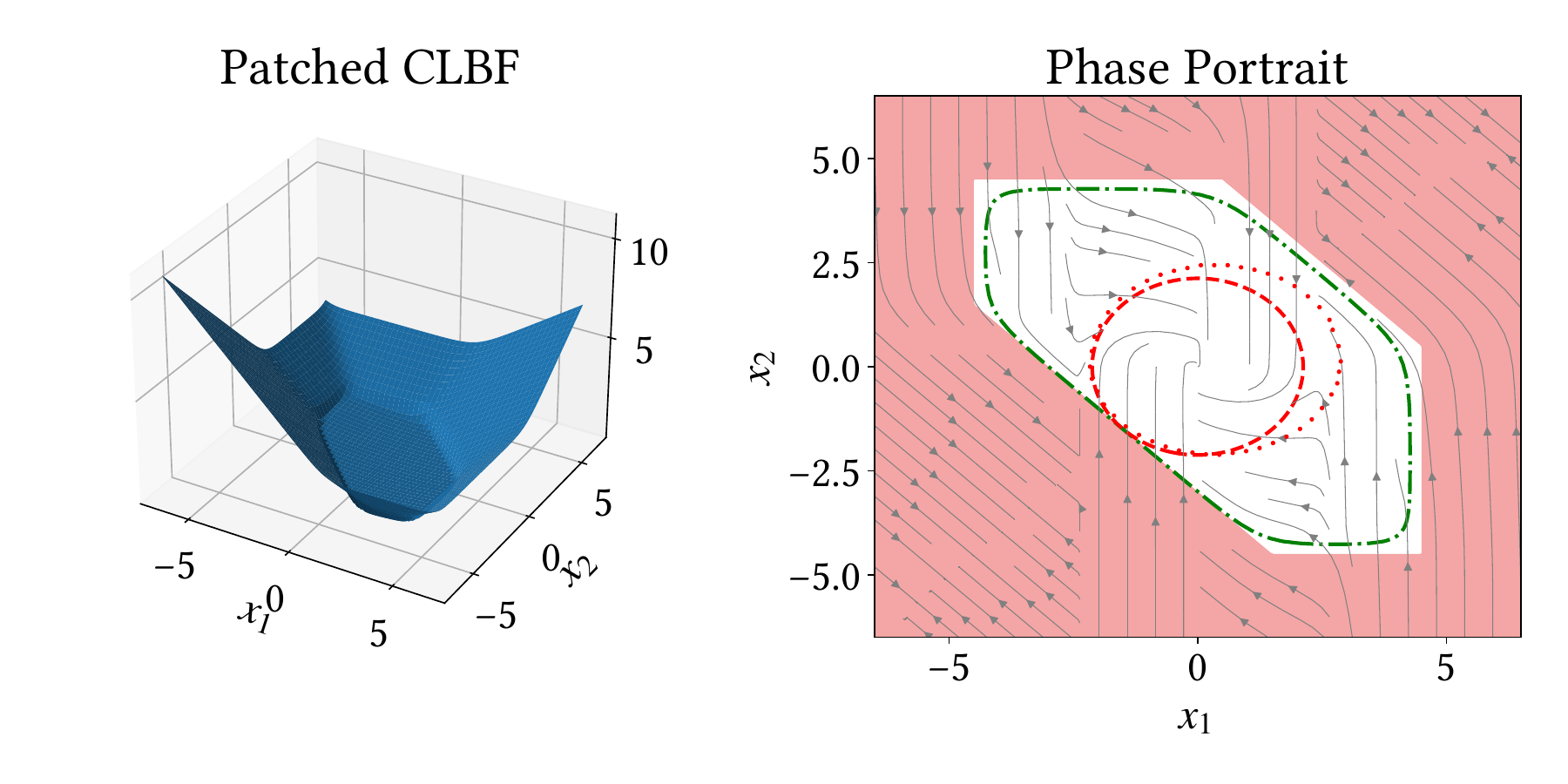}
    \caption{A smooth control Lyapunov-barrier function patched from a strictly compatible CBF-CLF pair for Example \ref{ex:2d_nonlinear_verified}. The green dash-dotted line represents the formally verified safe stabilization region, relative to the unsafe region (shaded in light red). For comparison, the best verified safe controllable region obtained using compatible SOS CBF-CLF~\cite{dai2024verificationv1} is shown in dotted red, while the dashed red curve depicts the largest safe stabilization region certified by a quadratic CLF.}
    \label{fig:2d_nonlinear_verified}
\end{figure}

\begin{example}\label{ex:2d_linear}
We consider the nonlinear control-affine system (\ref{eq:sys}), taken from \cite{dai2024verificationv1}, with 
\[
f(x)=\begin{bmatrix}x_1+x_2\\ -2x_1+3x_2 \end{bmatrix},
\quad
g(x)=\begin{bmatrix} 1 & 0 \\ 0 & 1 \end{bmatrix},
\]
and domain $[-10,10]^2$. The unsafe set is given by the single half-space $h_1(x):= -2 - x_1 - x_2 \le 1$, together with the box constraints defining the domain. We can verify strict compatibility of the softmax barrier function $h=h_{\mathrm{sm}}$ ($\tau=10.0$) with a quadratic CLF $V$. The resulting patched CLBF using (\ref{eq:patched_W}) and Theorem \ref{thm:patch} is depicted in Fig. \ref{fig:2d_linear}. 
\end{example}

\begin{figure}[h!]
    \centering
    \includegraphics[width=\linewidth]{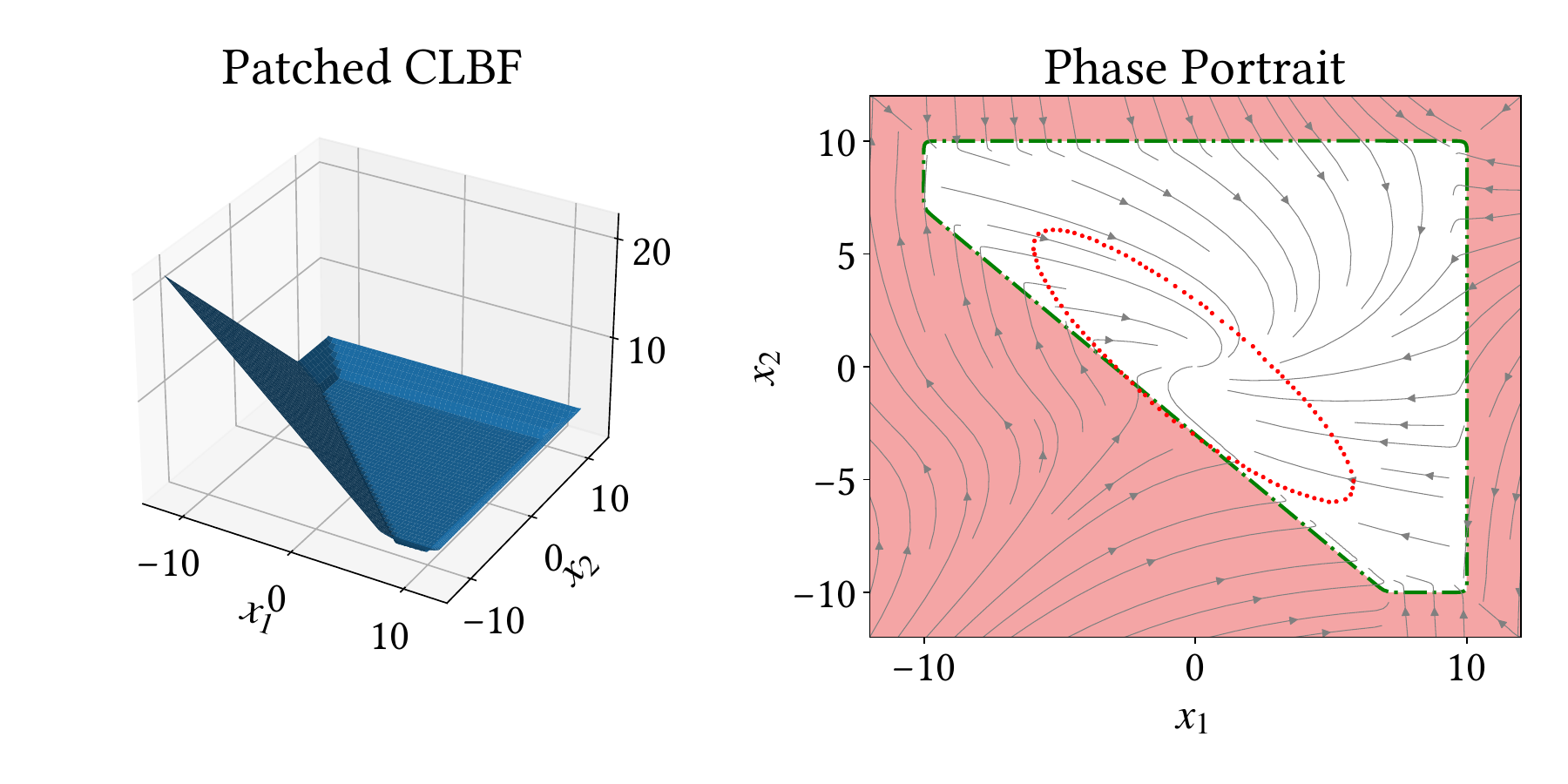}
    \caption{A smooth control Lyapunov-barrier function patched from a strictly compatible CBF-CLF pair for Example \ref{ex:2d_linear}. The green dash-dotted line represents the formally verified safe stabilization region, relative to the unsafe region (shaded in light red). For comparison, the best verified safe controllable region obtained using compatible SOS CBF-CLF~\cite{dai2024verificationv1} is shown in dotted red.}
    \label{fig:2d_linear}
\end{figure}

\begin{example}\label{ex:3d_power}
We consider the case study of a power converter, taken from \cite{schneeberger2024advanced} and also used in \cite{dai2024verificationv1}, which is a nonlinear control-affine system~\eqref{eq:sys} with
\[
f(x)=\begin{bmatrix}
-0.05x_1 - 57.9x_2 + 0.00919x_3\\[2pt]
1710x_1 + 314x_3\\[2pt]
-0.271x_1 - 314x_2
\end{bmatrix},
\]
\[
g(x)=\begin{bmatrix}
0.05 - 57.9x_2 & -57.9x_3\\[2pt]
1710 + 1710x_1 & 0\\[2pt]
0 & 1710 + 1710x_1
\end{bmatrix},
\]
on the domain $[-2,2]^3$. The unsafe set is specified by
\[
x_1 \le 0.2, 
\quad x_1 \ge -0.8, 
\quad (x_2-0.001)^2 + x_3^2 \le 1.2^2,
\]
together with the box constraints defining the domain. Here, the state vector $x=(x_1,x_2,x_3)^\top$ corresponds to the measured quantities $(v_{dc},i_d,i_q)^\top$, where $v_{dc}$ is the DC-link voltage and $i_d,i_q$ are the direct- and quadrature-axis currents in the synchronous $dq$-frame \cite{schneeberger2024advanced}. We construct a softmax barrier function $h=h_{\mathrm{sm}}$ with $\tau=3.1$ and verify strict compatibility with a quadratic CLF $V$ on the boundary band. The strict barrier condition was certified in less than $0.01$s, and strict compatibility in $0.59$s. The patched CLBF obtained from~\eqref{eq:patched_W} using Theorem~\ref{thm:patch} is depicted in Fig.~\ref{fig:3d_power}.
\end{example}

\begin{figure}[h!]
    \centering
    \includegraphics[width=\linewidth]{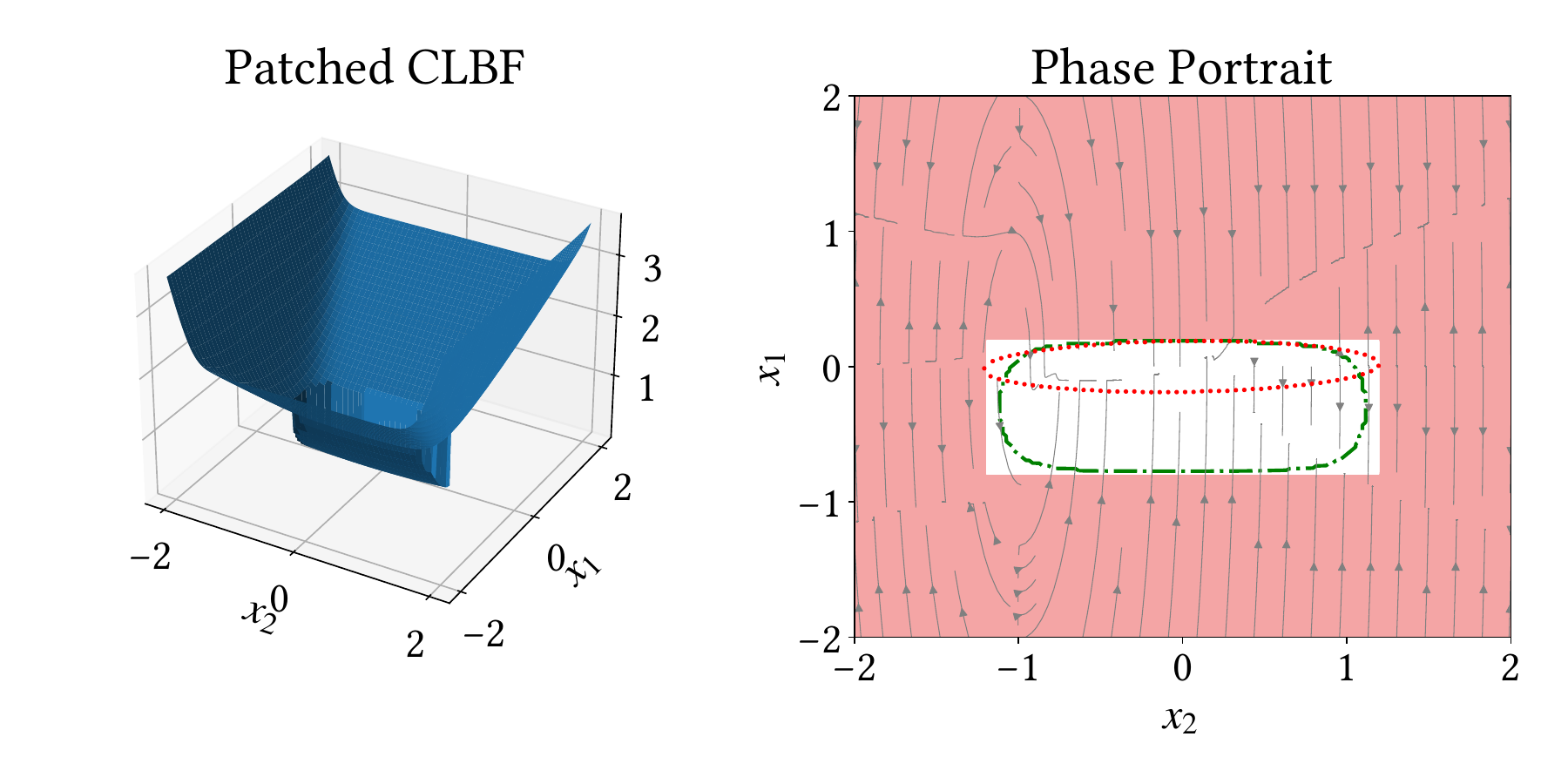}
    \caption{A patched control Lyapunov-barrier function for Example~\ref{ex:3d_power}. 
    The green dash-dot curve represents the formally verified safe stabilization region, 
    relative to the unsafe region (shaded in light red). For comparison, the best reported SOS CBF-CLF result \cite{dai2024verificationv1} is shown in dotted red.}
    \label{fig:3d_power}
\end{figure}

\section{Conclusions}

We developed a computational framework for constructing a single smooth Lyapunov function that certifies both stability and safety, enabled by softmax barrier relaxations, strict compatibility verification of CBF-CLF pairs, and explicitly defined smooth patching. Formal guarantees were established using $\delta$-complete SMT solvers, and the method was demonstrated on benchmark systems and compared with an alternative sum-of-squares approach. Future work will focus on incorporating input constraints and extending the approach with compositional verification techniques to scale to high-dimensional systems.

\bibliographystyle{plain}        %
\bibliography{acc26} 

\end{document}